\documentclass[conference]{IEEEtran}
\IEEEoverridecommandlockouts
\usepackage[utf8]{inputenc}
\usepackage[T1]{fontenc}
\usepackage{url}
\usepackage{ifthen}
\usepackage{cite}
\usepackage[cmex10]{amsmath} % Use the [cmex10] option 
\usepackage{amsmath, amssymb, amsthm}
\usepackage{bbm}
\usepackage{graphicx}
\usepackage[short]{optidef}
\usepackage{hyperref}
\usepackage{color}
\usepackage{algorithm, algpseudocode}

\newtheorem{lemma}{Lemma}

\begin{document}

\title{Optimization of Sparse VLSF Codes for Short-Packet Transmission via Saddlepoint Methods
\thanks{
This work is supported in part by the French National Agency for Research (ANR) through the project ANR-23-CMAS-0023 of the RIS3 program and the project ANR-22-PEFT-0010 of the France 2030 program PEPR Réseaux du Futur; and in part by the Agence de l'innovation de défense (AID) through the project UK-FR 2024352.}
}

\author{%
 \IEEEauthorblockN{Guodong Sun\IEEEauthorrefmark{1}, Samir M. Perlaza\IEEEauthorrefmark{1}\IEEEauthorrefmark{2}\IEEEauthorrefmark{3}, Philippe Mary\IEEEauthorrefmark{4}, Jean-Marie Gorce\IEEEauthorrefmark{5}}

 \IEEEauthorblockA{ Emails: \{guodong.sun, samir.perlaza, jean-marie.gorce\}@inria.fr,  philippe.mary@insa-rennes.fr
             }
 
 \IEEEauthorblockA{\IEEEauthorrefmark{1} Centre Inria d’Université Côte d’Azur, Inria, Sophia Antipolis, France
             }
     \IEEEauthorblockA{\IEEEauthorrefmark{2} Laboratoire GAATI, Université de la Polyn\'{e}sie fran\c{c}aise, Fa`a`\={a}, French Polynesia
             } \IEEEauthorblockA{\IEEEauthorrefmark{3} ECE Dept, Princeton University, Princeton, 08544 NJ, USA
             }         
  \IEEEauthorblockA{\IEEEauthorrefmark{4} Univ. Rennes, INSA, CNRS, IETR UMR 6164 F-35000, Rennes, France    }  
                 	
  \IEEEauthorblockA{\IEEEauthorrefmark{5} Inria, INSA Lyon, CITI Laboratory, UR3720, Villeurbanne, France
             }
}

% \author{%
%  \IEEEauthorblockN{Guodong Sun\IEEEauthorrefmark{1}, Samir M. Perlaza\IEEEauthorrefmark{1}, Philippe Mary\IEEEauthorrefmark{2}, Jean-Marie Gorce\IEEEauthorrefmark{3}}
 
%  \IEEEauthorblockA{\IEEEauthorrefmark{1} INRIA Université Côte d'Azur,
%              2004 Route des Lucioles, 06560 Valbonne, France\\
%              % email: guodong.sun|samir.perlaza@inria.fr
%              }
             
%   \IEEEauthorblockA{\IEEEauthorrefmark{2} INSA Rennes,
%   20 Avenue des Buttes de Coesmes, 35700 Rennes, France
%              % email: philippe.mary@insa-rennes.fr
%              }  
                 	
%   \IEEEauthorblockA{\IEEEauthorrefmark{3} INRIA Lyon,
%   56 Boulevard Niels Bohr, 69100 Villeurbanne, France
%              % email: jean-marie.gorce@inria.fr
%              }
% }
\maketitle

%\thanks{ This work has received a French government support granted to the Cominlabs excellence laboratory and managed by the National Research Agency in the ``Investing for the Future'' program under reference ANR-10-LABX-07-01 and also by the program \emph{France 2030} under reference ANR-23-CMAS-0001. This work was also funded by the Brittany region.}

\begin{abstract}
In this work, we present an optimization framework for sparse variable-length stop-feedback (VLSF) codes based on a saddlepoint approximation, which jointly optimizes the decoding configuration parameters.
Thanks to the analytical tractability of the saddlepoint approximation, the framework enables efficient gradient-based optimization of such parameters for common memoryless channels, including the additive white Gaussian noise, binary symmetric, and binary erasure channels.
We further propose a refined decoding rule that extends the conventional fixed-threshold rule and leads to a tighter achievability bound. 
Numerical results demonstrate that our framework provides near-optimal decoding configurations at low computational cost. 
Moreover, the results from our refined rule demonstrate that the fixed-threshold decoding rule is restrictive and that achievability bounds can be further tightened.
\end{abstract}

\section{Introduction}

Variable-length stop-feedback (VLSF) codes, in which transmission continues until an acknowledgment is received, are central to many feedback schemes and play an important role in channel coding~\cite{forney1968exponential, polyanskiy2011feedback}. 
For short-packet communication, VLSF codes can improve the capacity-achieving (first-order) rate and eliminate the finite-blocklength (second-order) penalty when stop feedback is available after every channel use and no delay constraints are imposed~\cite{polyanskiy2011feedback}. 
In practice, however, these gains may be restricted by feedback limitations, such as delay constraints~\cite{polyanskiy2011feedback}. 
For instance,  \cite{altuug2015variable} shows that probabilistic delay constraints allow for second-order improvements but do not eliminate the penalty. 
Similarly, studies of noisy feedback channels~\cite{ostman2020short} suggest that, in such settings, fixed-length codes without feedback may be preferable.

Practical VLSF codes are constrained by the cost of reserving feedback channel resources for possible transmission after each channel use, even when the resources are idle. 
This motivates sparse VLSF codes, which perform decoding and send stop feedback only at a small number of instants~\cite{williamson2015variable}. 
These schemes are widely implemented through hybrid automatic repeat request with incremental redundancy~\cite{dahlman20205g}. 
While \cite{williamson2015variable, kim2015variable} examine periodic sparse VLSF codes, more recent studies~\cite{vakilinia2016optimizing, yavas2023variable, yang2022vlsf, yang2022incremental} focus on optimizing the decoding instants.
These works show that the achievability bound of sparse VLSF can approach that of dense VLSF with only a few well-chosen decoding instants. 
In~\cite{vakilinia2016optimizing}, random stopping times are approximated to be Gaussian-distributed, and a sequential differential optimization (SDO) algorithm is proposed for optimizing decoding instants.
Building on this approach, the authors of~\cite{yavas2023variable} derive non-asymptotic achievability bounds for moderate blocklengths, while leaving the development of tight bounds for short packets as an open challenge.

To address the  short-packet regime, the authors of~\cite{yang2022vlsf, yang2022incremental} formulate achievability as a two-step optimization: first optimizing the decoding instants using a refined SDO algorithm; then optimizing the decoding threshold through a search. 
To apply the SDO algorithm, the information-density distribution for the binary-input additive white Gaussian noise (BI-AWGN) channel is approximated using Edgeworth~\cite{hall2013bootstrap} and Petrov~\cite{petrov2012sums} expansions, with different approximations employed across different regimes.
For the binary symmetric (BSC) and binary erasure (BEC) channels, the information density distributions are computed exactly. 
While this two-step approach is effective, the search step incurs nontrivial computational costs. This motivates a unified framework for joint optimization of the decoding instants and threshold.
The main contributions are summarized as follows:
\begin{itemize}
	\item We propose a refined decoding rule that anticipates the final decoding attempt by employing maximum-likelihood decoding, with the achievability bound determined by fixed-blocklength bounds~\cite{polyanskiy2010channel}. This yields a tighter bound than schemes relying solely on fixed thresholds.
    \item A saddlepoint approximation~\cite{butler2007saddlepoint} is proposed to characterize the probability distribution of the information-density for  a variety of channels, e.g. AWGN, BSC, and BEC. Analytical expressions that support efficient gradient-based optimization are then obtained and can be solved via a gradient-based mixed-integer nonlinear programming (MINLP) solver, e.g., Juniper~\cite{kroger2018juniper}.
	\item The proposed framework yields optimal or near-optimal decoding instants, substantially improving both optimization efficiency and generality across different channels. 
\end{itemize}
The remainder of this paper is organized as follows. Section~\ref{sec:problem} presents the problem and saddlepoint approximation. Section~\ref{sec:optimization} details the optimization method. Section~\ref{sec:results} discusses numerical results, and Section~\ref{sec:conclusion} concludes our contribution.

\section{System Model}\label{sec:problem}

\subsection{Sparse VLSF codes}

Consider a memoryless channel $P_{Y|X}$ with input $X\sim P_X$ and output $Y$ taking values in $\mathcal{X}$ and $\mathcal{Y}$, respectively.
For a blocklength $n\in \mathbb{N}_{>0}$, the input-output sequence is denoted by 
\begin{equation}
	({x}^n,{y}^n) = \big( (x_1, y_1), \dots, (x_n, y_n) \big) \in \mathcal{X}^n\times  \mathcal{Y}^n.
% \vspace{-0.1cm}
\end{equation}
The information density function is defined as 
\begin{equation}
	\imath_n: 
	\begin{function}
		{\mathcal{X}^n\times \mathcal{Y}^n}{ \mathbb{R}}{({x}^n, {y}^n)}{\log \frac{{\mathrm{d}}P_{Y^n|X^n=x^n}}{{\mathrm{d}}P_{Y^n}}(y^n),}
	\end{function}
\vspace{-0cm}
\end{equation}
where $P_{Y^n|X^n} = (P_{Y|X})^n$ and $P_{Y^n}=P_Y^{n}$ for the memoryless channel, with $P_Y$ the marginal on $\mathcal{Y}$ induced by $P_{Y|X}P_X$. 
We assume that the Radon–Nikodym derivative $\frac{{\mathrm{d}}P_{Y|X=x}}{{\mathrm{d}}P_Y}$ exists for $P_X$-almost every $x$ and the information density function can be written as 
\vspace{-0cm}
\begin{equation}
    \imath_n ({x}^n, {y}^n) = \sum_{i=1}^n\log\frac{{\mathrm{d}}P_{{Y}|{X}={x_i}}}{{\mathrm{d}}P_{{Y}}}(y_i).
\vspace{-0cm}
\end{equation}
When the blocklength $n$ is clear from context, we omit the subscript and write $\imath ({x}^n, {y}^n)$.

A sparse VLSF code is specified by $(\mathcal{T}, M, \epsilon)$, in which $\mathcal{T}\triangleq \{n_1, \dots, n_t\}\in \mathbb{N}_{>0}$ with $t=|\mathcal{T}|$ is the set of admissible decoding instants, $M$ is the number of messages, and $\epsilon$ is the average error probability constraint. 
A message, denoted by $W$, is uniformly distributed over the message set $\mathcal{M}\triangleq\{1, \ldots, M\}$.
A random codebook with independent and identically distributed (i.i.d.) codewords of length $n_t$ is a tuple
\begin{equation}
    \mathcal{C} \triangleq \big({X}^{n_t}_1,  \dots,  {X}^{n_t}_M \big), \quad  {X}^{n_t}_m\sim P_X^{ n_t}, \quad m\in \mathcal{M}.
\end{equation}

Let $\gamma>0$ denote a threshold. 
At each decoding instant $n \in \mathcal{T}$, the decoder evaluates the information density for all codewords and proceeds as follows: (1) it continues decoding if no codeword's information density exceeds $\gamma$; (2) it makes a decision if exactly one codeword's information density exceeds $\gamma$, in which case the index of this codeword is declared; (3) it declares an error if more than one codeword's information density exceeds $\gamma$. 
For each $m\in \mathcal{M}$, define the stopping time of the $m$-th codeword as: 
% \vspace{-0.3cm}
\begin{equation}
    \tau_m \triangleq \min\{n\in \mathcal{T}: \imath ({X}^{n}_m, {Y}^n)\geq \gamma\}.
% \vspace{-0.1cm}
\end{equation}
The stopping time of the decoder is
\vspace{-0cm}
\begin{equation}
    \tau^* \triangleq \min_{m\in \mathcal{M}} \tau_{m},
\vspace{-0cm}
\end{equation}
with the convention $ \tau^* = n_t$ if no codeword reaches the threshold by the final decoding instant $n_t$.
The achievable rate of the code is defined as the ratio of the message size (in bits) to the expected minimum stopping time~\cite{polyanskiy2011feedback}
\vspace{-0cm}
\begin{equation}
    r = {\log_2(M)/}{\mathbb{E}[ \tau^* ]}. 
\vspace{-0cm}
\end{equation}

Yavas \emph{et al}. \cite[Theorem 1]{yavas2023variable} established an achievability bound for sparse VLSF codes.
Building on this result, Yang \emph{et al}. \cite{yang2022incremental} obtained the maximal achievability bound by solving the following optimization problem, referred to as P.~\ref{opt:literature}
\begin{mini!}
	{\gamma, \mathcal{T}}{n_1 + \sum_{j=1}^{|\mathcal{T}|-1}(n_{j+1}-n_{j})\mathbb{P}[\imath({X}^{n_j}_1,{Y}^{n_j})< \gamma], }{\label{opt:literature}}{ \label{obj:literature} }
	\addConstraint{\mathbb{P}[\imath({X}^{n_t}_1,{Y}^{n_t})< \gamma] + (M-1)e^{-\gamma}}{\leq \epsilon\label{con:decoding_error_literature}},
    % \vspace{-0.1cm}
\end{mini!}
where \eqref{obj:literature} is an upper bound on the expected decoding time $\mathbb{E}[\tau^*]$ under the threshold decoding rule. The error probability constraint~\eqref{con:decoding_error_literature} is determined by the final decoding attempt.
Specifically, $(M-1)e^{-\gamma}$ is an upper bound on the probability that any non-transmitted codeword exceeds the threshold before the transmitted one (false-alarming), while $\mathbb{P}[\imath ({X}^{n_t}_{1},{Y}^{n_t}) < \gamma]$ represents the probability that the information density of the transmitted codeword, assuming it is $X_1^n$, fails to reach the threshold at the final decoding attempt (missed-detection). 
This formulation also accounts for intermediate decoding attempts, since the false-alarming term $(M-1)e^{-\gamma}$ is already included in the constraint for each prior decoding attempt. 
Note that this error probability constraint aligns with the one used to determine an achievability bound for the fixed-blocklength codes, i.e., Shannon's bound~\cite[Theorem 2]{polyanskiy2010channel}.%\cite{shannon1957certain}.

Since no further transmissions occur after the final decoding attempt,
%the decoder can employ a different decision rule at that stage instead of thresholding.
we propose a two-stage decoding strategy: during the intermediate decoding attempts, the decoder applies the threshold-based rule, while at the final decoding attempt, it selects the codeword with the maximal information density. 
For this final decoding attempt, the average error probability $\epsilon^*_{fb}(n_t, M)$ for a fixed blocklength $n_t$ and codebook size $M$ was investigated via random coding in \cite{polyanskiy2010channel}. 
We now formulate the optimization problem for the refined decoding rule, which we refer to as P.~\ref{opt:refined}:
\begin{mini!}
	{\gamma, \mathcal{T}}{n_1 + \sum_{j=1}^{|\mathcal{T}|-1}(n_{j+1}-n_{j})\mathbb{P}[\imath({X}^{n_j}_{1},{Y}^{n_j})< \gamma],}{\label{opt:refined} }{}
	\addConstraint{ (M-1)e^{-\gamma}+\epsilon^*_{fb}(n_t, M)}{\leq \epsilon \label{con:decoding_error}}
\end{mini!}

As shown in P.~\ref{opt:literature} and P.~\ref{opt:refined}, the key component is the CDF $\mathbb{P}[\imath({X}^{n_j}_1, {Y}^{n_j})< \gamma]$ at any decoding instant $n_j\in \mathcal{T}$. For notational convenience, for $n\in \mathbb{N}_{>0}$ we define 
\vspace{-0cm}
\begin{equation}
    S_{n} \triangleq \imath({X}^{n}_1, {Y}^n) \stackrel{}{=} \sum_{i=1}^{n}\imath({X}_{1,i}, Y_i) = \sum_{i=1}^{n}Z_i,
\vspace{-0cm}
\end{equation}
where $Z_i\triangleq \imath({X}_{1,i},{Y})$ denotes the single-letter information density of the transmitted codeword $X_1^{n_t}$ at time $i$.
We also write $Z\triangleq \imath({X}_{1},{Y})$ for a generic single-letter information density.

\subsection{Saddlepoint approximation}

Next, we introduce the saddlepoint approximation of $\mathbb{P}[S_n < \gamma]$.
Let $K_Z(s)$ denote the cumulant generating function (CGF) of $Z$. 
For $S_n = \sum_{i=1}^{n}Z_i$, the CGF is $K_{S_n}(s) = nK_Z(s)$ since $Z_i$ are i.i.d.
The saddlepoint $\hat{s}$ is then computed by solving $K_{S_n}'(\hat{s})=\gamma$, where $K_{S_n}'(s)$ is the first derivative of the CGF.

If $Z$ has a density, the CDF of $\mathbb{P}[S_n < \gamma]$ for each $n\in \mathbb{N}_{>0}$ can be accurately approximated using the Lugannani-Rice saddlepoint approximation \cite[1.2.1]{butler2007saddlepoint}\cite{lugannani1980saddle}
% \vspace{-0cm}
\begin{equation}\label{eq:saddle_continuous}
	\mathbb{P}[S_n < \gamma] \approx \begin{cases}
		\Phi(\hat{w})+\phi(\hat{w})(1/\hat{w}-1/\hat{u}) & \text{if}~\gamma\neq \mathbb{E}[S_n],\\
		\frac{1}{2} + \frac{K'''(0)}{6\sqrt{2\pi}K''(0)^{3/2}} & \text{if}~\gamma =  \mathbb{E}[S_n],
	\end{cases}
% \vspace{-0.1cm}
\end{equation} 
where $\hat{w} = \text{sgn}(\hat{s})\sqrt{2(\hat{s}\gamma-K(\hat{s}))}$ and $\hat{u} = \hat{s}\sqrt{K''(\hat s)}$. The functions $\phi$ and $\Phi$ denote the standard normal probability density and CDF, respectively, and $\text{sgn}$(.) is the sign function. 

For discrete $Z$, a first-order continuity correction~\cite[1.2.3]{butler2007saddlepoint} is applied to account for the lattice structure of $S_n$.  
Let $k$ denote the smallest attainable lattice point above $\gamma$ and the saddlepoint $\hat{s}$ is then computed by solving $K_{S_n}'(\hat{s}) = k$. $\mathbb{P}[S_n < \gamma] = \mathbb{P}[S_n < k]$ is approximated by the discrete saddlepoint formula:
% \vspace{-0.3cm}
\begin{equation}\label{eq:discrete_saddle_point}
	\mathbb{P}[S_n < \gamma]\! \approx\! \begin{cases} 
		\Phi(\hat{w})\!+\!\phi(\hat{w})(1/\hat{w}\!-\!1/\bar{u}) & \!\!\text{if}~k\neq  \mathbb{E}[S_n],\\
		 \frac{1}{2} \!+ \!\frac{1}{\sqrt{2\pi}}\big(\frac{K'''(0)}{6K''(0)^{3/2}}\! - \!\frac{1}{2\sqrt{ K''(0)}}  \big) &\!\! \text{if}~k =   \mathbb{E}[S_n],
	\end{cases}
\end{equation} 
where $\hat{w}$ is defined as in the continuous case, and the curvature correction term is given by $\bar{u} = (1-e^{-\hat{s}})\sqrt{K''(\hat s)}$. 

For both the continuous and discrete saddlepoint approximations, when $\gamma\neq \mathbb{E}[S_n]$ in \eqref{eq:saddle_continuous}  or $k\neq\mathbb{E}[S_n]$ in \eqref{eq:discrete_saddle_point}, the quantity $\Phi(\hat{w})+\phi(\hat{w})(1/\hat{w}-1/\hat{u})$ is smooth with respect to the saddlepoint $\hat{s}$, since $\Phi$, $\phi$, $\hat w$, and $\hat u$ are smooth functions of $\hat{s}$.
The special cases at $\gamma=\mathbb{E}[S_n]$  or $k=\mathbb{E}[S_n]$  are defined by continuation, and hence the saddlepoint approximation is defined everywhere.

\section{Decoding Schedule Optimization}\label{sec:optimization}

In this section, we show that for the AWGN, BSC, and BEC channels, the CDF $\mathbb{P}[S_n < \gamma]$ is continuous w.r.t. the channel parameters, enabling the application of gradient-based methods to solve the optimization problems.

\subsection{AWGN channel}

For the AWGN channel $Y=X+N$, where $X \sim \mathcal{N}(0, p_0)$ and $N \sim \mathcal{N}(0, 1)$, and where $p_0$ denotes the normalized SNR\footnote{It is known that i.i.d. Gaussian signaling exhibits a larger dispersion in the finite-blocklength regime than the dispersion-optimal distribution~\cite{molavianjazi2014unified}. We adopt i.i.d. Gaussian signaling for analytical tractability, as it provides a closed-form CGF required for the saddlepoint approximation.}. Here, $\mathcal{N}$ denotes the normal distribution. 
The information density  is 
\vspace{-0cm}
\begin{equation}\label{eq:awgn}
Z = \frac{1}{2}\log\left(1+{p_0}\right)+\frac{1}{2}\left( \frac{Y^2}{p_0+1} - {(Y-X)^2}\right).
\end{equation}
Let  $\tilde{Z}=Z-\frac{1}{2}\log\left(1+{p_0}\right)$ denote the shifted information density of $Z$ and define $\tilde{S}_n = \sum_{i=1}^{n}\tilde{Z}_i$. The corresponding shifted threshold at time $n$ is:
\vspace{-0cm}
\begin{equation}\label{eq:tilde_gamma_awgn}
    \tilde\gamma_n = \gamma - \frac{n}{2}\log\big(1+{p_0}\big).
\vspace{-0cm}
\end{equation}
The following lemma provides a closed-form expression for the saddlepoint:
% \vspace{-0cm}
\begin{lemma}
The CGF of $\tilde{S}_n$ is given by 
% \vspace{-0cm}
\begin{equation}
	K_{\tilde{S}_n}(s)=-\frac{n}{2}\log\left(1-\frac{p_0s^2}{p_0+1}\right),
\vspace{-0cm}
\end{equation}
and the corresponding saddlepoint $\hat{s}$ is given by 
% \vspace{-0cm}
\begin{equation}\label{eq:tilde_s_awgn}
	\hat{s} = \frac{-n+\sqrt{n^2+4\tilde{\gamma}^2(p_0+1)/(p_0)}}{2\tilde{\gamma}}. 
% \vspace{-0cm}
\end{equation}
\end{lemma}
\begin{proof}
	See Appendix~\ref{app:AWGN}
\end{proof}
Substituting this closed-form saddlepoint into~\eqref{eq:saddle_continuous} allows $\mathbb{P}[S_n < \gamma]$ to be evaluated as a closed-form function with respect to $n$ and $\gamma$. 
\begin{figure}
	\includegraphics[width=\linewidth]{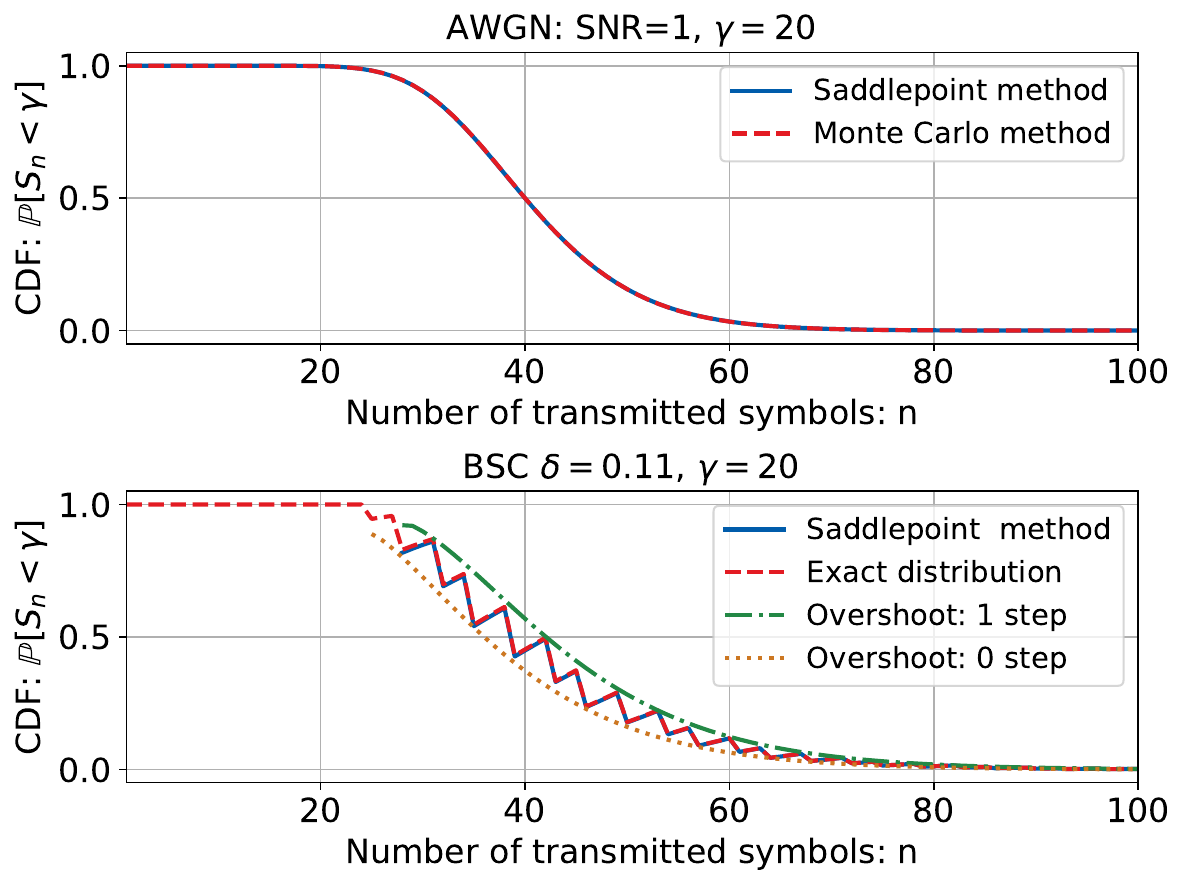}
	\caption{Saddlepoint approximation versus the true distribution. For the AWGN channel, the true distribution is estimated by Monte Carlo simulation, and for the BSC it is computed exactly. For the BSC with discrete $S_n$, overshoot at threshold crossings is between zero and one step; Saddlepoint approximations for both fixed overshoot extremes are shown.}
	\label{fig:Saddlepoint_matching}
\end{figure}
As illustrated in the upper subfigure of Fig.~\ref{fig:Saddlepoint_matching}, the tail distribution obtained via the saddlepoint approximation closely matches the true tail distribution computed through extensive Monte Carlo simulations ($10^6$ trials).  
This highlights both the accuracy and the inherent smoothness of the saddlepoint approximation.
This continuity allows the optimization problem to be solved using a standard gradient-based MINLP solver, which relaxes $n\in \mathbb{N}_{>0}$ to $n\in \mathbb{R}$ and then applies a local search.  
The corresponding results are presented in the next section.

\subsection{BSC or BEC}\label{sec:overshoot}

For a BSC with crossover probability $\delta\in (0, \frac{1}{2})$ and input $X \sim \text{Bernoulli}(\frac{1}{2})$, the information density of each transmitted symbol takes two values:
\begin{equation}\label{eq:info_density_BSC}
	Z =\begin{cases}
		\log 2(1-\delta) 	& \text{w.p.}~1-\delta,\\
		\log 2\delta & \text{w.p.}~\delta,
	\end{cases}
\vspace{-0cm}
\end{equation}
where “w.p.” means “with probability.”  
Define the shifted information density $\tilde{Z} \triangleq Z - \log 2\delta$. At time $n$,  $\tilde{S}_{n} = \sum_{i=1}^{n} \tilde{Z}_i = \sum_{i=1}^nZ_i -n\log 2\delta$. 
The CGF of $\tilde{S}_n$ is 
\begin{equation}
	K_{\tilde{S}_n}(s) = n\log\big(\delta + (1-\delta)e^{\log\frac{1-\delta}{\delta} s}\big). 
\end{equation}
Recall that $k$ is the smallest attainable lattice point exceeding $\gamma$. After shifting, the corresponding lattice point becomes $k - n \log 2\delta$.
Solving $K_{\tilde{S}_n}'(s) = k -n\log 2\delta $ gives a closed form expression for the saddlepoint:
\begin{equation}\label{eq:tilde_s_BSC}
	\hat{s} = \log\bigg(\frac{(k/n-\log 2\delta)\delta}{(\log\big(2(1-\delta)\big)-k/n)(1-\delta)}\bigg)\Big/\log\left(\frac{1-\delta}{\delta}\right).
\end{equation}
Substituting $\hat{s}$ into \eqref{eq:discrete_saddle_point} allows $\mathbb{P}[S_n < \gamma]$ to be expressed in closed-form as a function of $n$ and the overshooting lattice point $k$.

Similarly, for the BEC with erasure probability $\delta$ under the input $X \sim \text{Bernoulli}(\frac{1}{2})$, the information density $Z$ of each transmitted symbol is a discrete random variable:
\begin{equation}
	Z =\begin{cases}
		1 	& \text{w.p.}~1-\delta,\\
		0  & \text{w.p.}~\delta.
	\end{cases}
\end{equation}
which can be viewed as a special case of the BSC in \eqref{eq:info_density_BSC}, and we omit the derivation of the saddlepoint for BEC.

{For BSC and BEC, the cumulative information density $S_n$ is defined on a lattice.
Consequently, the smallest lattice point $k$ exceeding $\gamma$ results in an overshoot $D\triangleq k-\gamma$, where $D$ is nonnegative but smaller than one step.
For BSC with $\gamma\in \mathbb{R}$, we have $D\in \left[0, \log\big(\frac{1-\delta}{\delta}\big) \right)$.
To enable gradient-based MINLP optimization, we first apply a continuous correction for approximating the overshoot.
Since the overshoot is bounded by one step, the actual CDF lies between the CDF approximations without overshoot $\mathbb{P}[S_n<\gamma]$ and $\mathbb{P}\left[S_n<\gamma+\log\big(\frac{1-\delta}{\delta}\big)\right]$, as shown in the lower subfigure of Fig.~\ref{fig:Saddlepoint_matching}. 
This continuous correction ensures that the saddlepoint approximation is smooth in $n$ and $\gamma$, and hence suitable for gradient-based optimization. 
We first solve the continuous relaxation for both boundary values $D=0$ and $D=\log\big(\frac{1-\delta}{\delta}\big)$, following the approach used for AWGN channels. 
We then refine the solution via a local search within these decoding ranges to obtain the optimal parameters.}

\subsection{Instability of saddlepoint approximation}

A technical difficulty for applying the saddlepoint approximation, specifically for discrete channels, is the numerical instability near the mean, i.e., when $\gamma \approx \mathbb{E}[S_n]$~\eqref{eq:discrete_saddle_point}. 
In this region, the saddlepoint $\hat{s}$, obtained by solving, $K'(\hat{s})=\gamma$ is approximately zero.
Then, $\hat{w}=\text{sgn}(\hat{s})\sqrt{2(\hat s \gamma-K(\hat{s}))}$ and $\bar{u}={1-e^{-\hat{s}}}\sqrt{K''(\hat{s})}$ are also approximately zero. 
The difference $\frac{1}{\hat{w}} - \frac{1}{\bar{u}}$ in  \eqref{eq:discrete_saddle_point} is then numerically ill-conditioned, and the impact of approximating overshoot is huge. 
To address this, we apply a second-order Taylor expansion of the CGF $K(\hat{s})$ for $\hat{s}$ around 0: 
\vspace{-0cm}
\begin{equation}
	K(\hat s) = \mu \hat{s}+\frac{1}{2}\sigma^2\hat{s}^2 + O(\hat{s}^3),
\vspace{-0cm}
\end{equation}
where $\mu$ and $\sigma^2$ are the mean and variance of $S_n$. 
Then $K'(s)\approx \mu + \sigma^2\hat{s}$ and solving $K'(\hat s) = \gamma$ gives $\hat s\approx \frac{\gamma-\mu}{\sigma^2}$.
Substituting into $\hat{w}$, we have
\vspace{-0cm}
\begin{equation}
	\hat{w}=\sqrt{2(\hat s \gamma-K(\hat s))} 
	\approx\frac{\gamma-\mu}{\sigma}.
\vspace{-0cm}
\end{equation}
So that near the mean, the CDF of $S_n$ can be well approximated by $\Phi(\frac{\gamma-\mu}{\sigma})$. 
The CDF is then given by
\vspace{-0cm}
\begin{equation}
	\mathbb{P}[S_n< \gamma]	= \begin{cases}
		\Phi(\frac{\gamma-\mu}{\sigma}) & |\frac{\gamma-\mu}{\sigma}|\leq \epsilon_s \\
		\eqref{eq:discrete_saddle_point} & |\frac{\gamma-\mu}{\sigma}| > \epsilon_s,
	\end{cases}
\vspace{-0cm}
\end{equation}
where $\epsilon_s$ is a small empirically chosen threshold (typically $\epsilon_s=0.1$)
that separates the unstable central region from the tails, where the full saddlepoint approximation in~\eqref{eq:discrete_saddle_point} is used.
Note that the piecewise smooth approximation may introduce a discontinuity at the transition. 
This has little practical impact, as the optimizer still converges reliably given the narrowness of the corrected region, as confirmed in Section~\ref{sec:results}.
\section{Numerical Results}\label{sec:results}

We evaluate our framework by comparing the optimized decoding attempts obtained with our method against those from a brute-force search. 
The gradient-based MINLP problem is solved using the Juniper package in Julia~\cite{kroger2018juniper}.

\subsection{Evaluation of the optimization framework}
In the following, $t=3$ decoding attempts, message sizes between 30 and 120 bits, SNR $=1$, and $\epsilon=10^{-3}$ are considered, unless otherwise specified.
\paragraph{Computational efficiency}
For both AWGN and BSC channels, the gradient-based optimization is far more efficient than brute-force search: it produces results in under one second, whereas the brute-force approach can take several hours. 
Moreover, brute-force search is computationally feasible only for $t \leq 3$ due to the exponential growth of the search space, which highlights the advantage of the gradient-based method.

\paragraph{AWGN}
\begin{figure}
	\includegraphics[width=\linewidth]{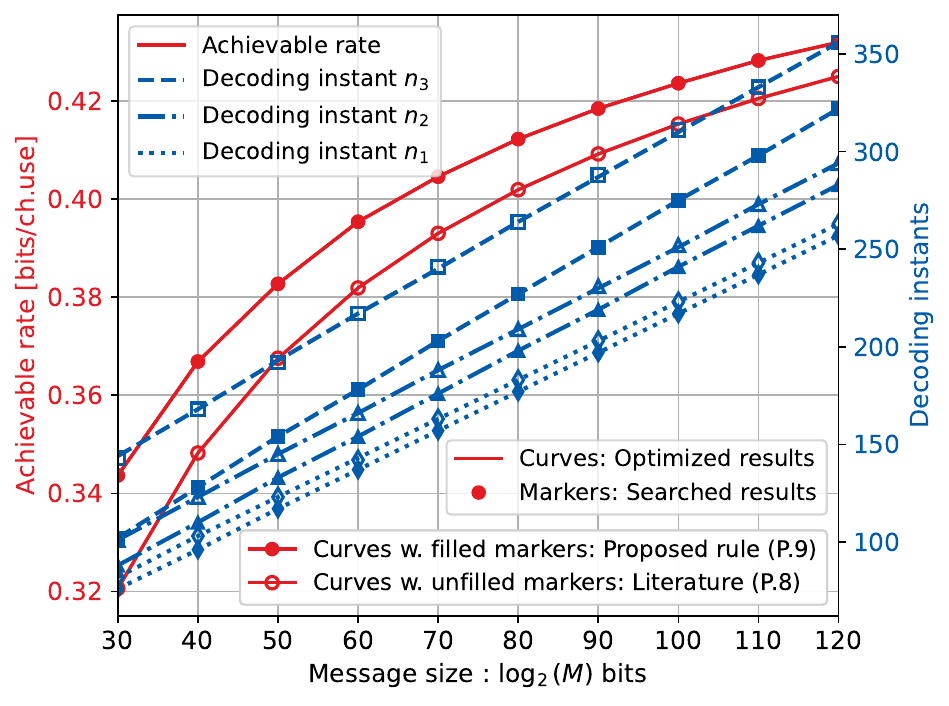}
	\caption{Achievable rate (left) and optimal decoding instants (right) versus message size (in bits) in an AWGN channel ($\text{SNR}=1$, $t=3$ attempts). 
 }
	\label{fig:match_awgn}
\end{figure}

Fig.~\ref{fig:match_awgn} presents the achievable rate (vertical left axis) and the corresponding decoding instants (vertical right axis) with respect to the message size. The gradient-based optimization (curves) closely matches the brute-force search (markers), confirming the effectiveness of our approach.

Compared with the rule in the literature (P.~\ref{opt:literature}), the proposed decoding rule (P.~\ref{opt:refined}) achieves higher rates, indicating a closer approach to the fundamental limit of sparse VLSF codes. 
This improvement is more significant for small message sizes, with nearly an 8\% gain at 30 bits, decreasing to approximately 2\% at 120 bits. Hence, the proposed rule accurately captures the performance limits for short-packet communication.
On the right axis, the optimized decoding instants under the proposed rule (filled markers) are consistently earlier than those from the literature rule (unfilled markers). The gap is especially evident at later decoding stages (e.g., $n_2, n_3$), suggesting that threshold-based decoding is restrictive in the later stages. Further work is needed to refine the decoding rule for stages that allow early stopping.

\paragraph{BSC}

\begin{figure}
\centering
	\includegraphics[width=\linewidth]{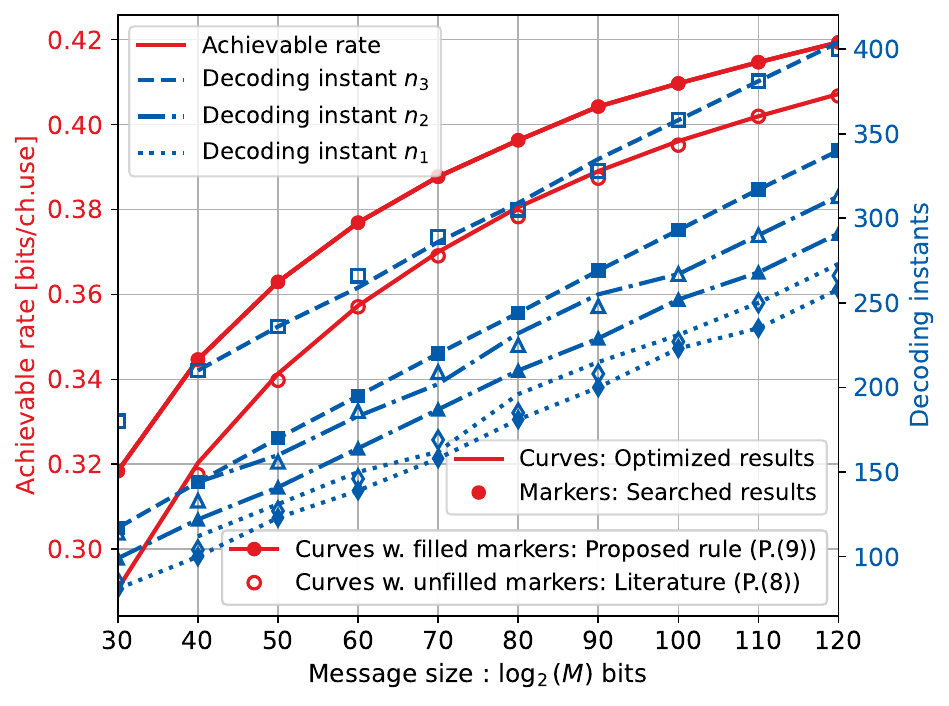}
	\caption{Achievable rate (left) and optimal decoding instants (right) versus message size (in bits) in a BSC ($\delta=0.11$, $t=3$ attempts). 
    }
	\label{fig:BSC_decoding_times}
\end{figure}

Fig.~\ref{fig:BSC_decoding_times} shows the results for a BSC with crossover probability $\delta=0.11$, corresponding to a Shannon capacity of approximately $0.5$ bit/s/Hz.

Fig.~\ref{fig:BSC_decoding_times} shows a similar advantage for the BSC as observed in the AWGN channel: the proposed decoding rule achieves higher rates as the decoding times are anticipated. 
This demonstrates that the saddlepoint approximation, combined with gradient-based optimization, provides a general approach applicable across different channels, enabling a systematic investigation of decoding rules.
However, as seen in comparison to brute-force search, gradient-based optimization does not guarantee globally optimal decoding schedules for the BSC, since gradients cannot fully capture the exact CDF of threshold crossing with overshoot (shown in the lower panel of Fig.~\ref{fig:Saddlepoint_matching}). 
Nevertheless, the results remain near-optimal; the achievable rates closely align with those from a brute-force search. Furthermore, while the decoding instants do not match perfectly as they do in the AWGN case, they remain close to the results obtained by brute-force search.

\subsection{Evaluation of the proposed achievable bound}
Next, we focus on the AWGN channel, which is of primary interest in communication problems, having noted that similar performance trends are observed for the BSC.

\begin{figure}
	\includegraphics[width=0.92\linewidth]{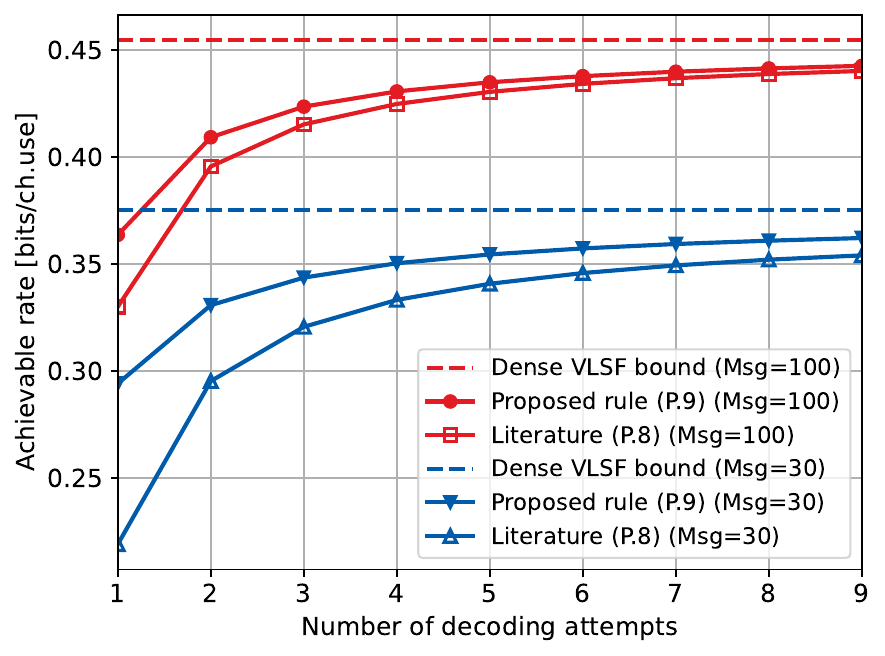}
    % \vspace{-0.3cm}
	\caption{Achievable rate of sparse VLSF codes versus number of decoding attempts.}
	\label{fig:rate_number_decoding_times}
    % \vspace{-0.5cm}
\end{figure}

\begin{figure}
	\includegraphics[width=0.93\linewidth]{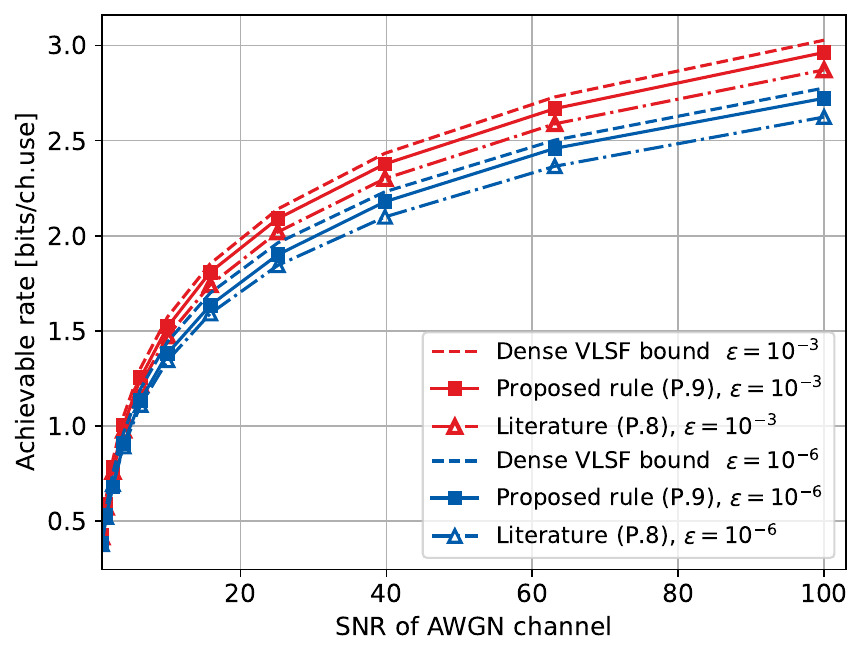}
        % \vspace{-0.3cm}
	\caption{Achievable rate with respect to SNR. The target decoding error is set as $\epsilon_0=\{10^{-3}, 10^{-6}\}$, respectively.}
	\label{fig:SNR}
        % \vspace{-0.5cm}
\end{figure}

Fig.~\ref{fig:rate_number_decoding_times} shows achievable rates versus the number of sparse decoding attempts for SNR$=1$ and message sizes of $\{30, 100\}$ bits. Note that the reference curve representing dense decoding after every received symbol is also reported~\cite{polyanskiy2011feedback}. 
The gap to the reference closes quickly, indicating that near-optimal performance can be achieved with only a few attempts. 
The proposed rule lies closer to the reference, achieving substantially better performance even with one or two attempts. 
As the number of decoding attempts increases, the gap narrows, showing that the refined rule is most significant for fewer attempts.

In general, the performance of finite blocklength codes depends on parameters such as message size, blocklength, target error probability, and channel conditions. The current framework allows us to analyze the interactions among these parameters. Fig.~\ref{fig:SNR} shows achievable rates for both decoding rules w.r.t. SNR for message size of 100 bits and $\epsilon=\{10^{-3}, 10^{-6}\}$. The gap between sparse and dense decoding widens at higher SNR, and the difference between the literature rule and the proposed rule also grows. 
Nevertheless, overall trends remain almost identical across $\epsilon$, indicating that SNR has a greater impact than the target error on sparse VLSF performance.

\section{Conclusion}\label{sec:conclusion}

This paper presents a refined decoding rule for the final decoding attempt in VLSF codes, as well as an optimization framework for decoding instants based on the saddlepoint approximation of the information density.
The resulting gradient-based framework provides a computationally efficient approach to optimizing decoding schedules and characterizing the performance of sparse VLSF codes. 
Moreover, the proposed rule significantly tightens the achievability bound.
Future work could investigate alternatives to the potentially restrictive fixed-threshold decoding scheme used in the intermediate decoding steps to achieve further rate improvements.

% \section{Acknowledgment}
% This work is supported in part by the French National Agency for Research (ANR) through the project ANR-23-CMAS-0023 of the RIS3 program and the project ANR-22-PEFT-0010 of the France 2030 program PEPR Réseaux du Futur; and in part by the Agence de l'innovation de défense (AID) through the project UK-FR 2024352.
\appendix
\subsection{Computing saddlepoint for the AWGN channel}\label{app:AWGN}

The shifted random variable $\tilde{Z}$ is given by
\begin{equation}\label{Eq:def_U}
	\tilde{Z} = \frac{1}{2}\left( \frac{Y^2}{p_0+1} - {(Y-X)^2}\right) = \frac{1}{2}\left( \frac{(X+N)^2}{p_0 + 1} - {N^2}\right).
\end{equation}
Define $V=\frac{X+N}{\sqrt{p_0 + 1}}\sim \mathcal{N}(0, 1)$, and $W={N}\sim \mathcal{N}(0, 1).$
Since both depend on the same noise variable $N$, $V$ and $W$ are correlated standard normal variables with covariance $\rho\triangleq\text{cov}(V, W)=\sqrt{\frac{1}{p_0+1}}$. 
% 
% \begin{equation}
% 	\rho\triangleq\text{cov}(V, W)=\mathbb{E}[VW] = \frac{\mathbb{E}[XN]+\mathbb{E}[N^2] }{\sqrt{p_0+1}} = \sqrt{\frac{1}{p_0+1}}.
% \end{equation}
From \eqref{Eq:def_U}, $\tilde{Z}$ is a quadratic form
\begin{equation}
	\tilde{Z}=\frac{1}{2}(V^2-W^2) = \begin{pmatrix}
		V&W
	\end{pmatrix}\begin{pmatrix}
		\frac{1}{2}&0\\0&-\frac{1}{2}
	\end{pmatrix}\begin{pmatrix}
		V\\W
	\end{pmatrix}.
\end{equation}
The moment generating function (MGF) of such a quadratic form is given by~\cite[Theorem~5.2b]{rencher2008linear}
% \vspace{-0.3cm}
\begin{equation}\label{Eq:constraint_AWGN}
\begin{aligned}
		M_{\tilde{Z}}(s) =& \left(\det \left(I-2s\begin{pmatrix}
		1&\rho\\ \rho& 1
	\end{pmatrix}\begin{pmatrix}
		\frac{1}{2}&0\\0&-\frac{1}{2}
	\end{pmatrix}\right)  \right)^{-1/2}\\
	=&\frac{1}{\sqrt{1-s^2(1-\rho^2)}},
\end{aligned}
\end{equation}
which exists for values of $s$ within the region of convergence $s < \sqrt{{(p_0+1)}/{p_0}}.$ 
Since $\tilde{S}_n=\sum_{i=1}^n \tilde{Z}_i$, the CGF of $\tilde{S}_n$ is 
\begin{equation}
	K_{\tilde{S}_n}(s)= \log \big(M_{\tilde{Z}}(s)\big)^n = -\frac{n}{2}\log\left(1-\frac{p_0s^2}{p_0+1}\right).
\end{equation}
To find the saddlepoint, we have $K'(s)= \frac{np_0s}{(p_0+1)-p_0s^2}$.
The saddlepoint $\hat{s}(\tilde{\gamma})$ is the solution to $K'(s)=\tilde{\gamma}$, which gives the unique root in the region of convergence
 \begin{equation}
	\hat{s}(\tilde{\gamma}) = \frac{-n+\sqrt{n^2+4\tilde{\gamma}^2(p_0+1)/(p_0)}}{2\tilde{\gamma}}.
\end{equation}
% where the other root is not in the region of convergence for~\eqref{Eq:constraint_AWGN}.

\bibliographystyle{ieeetr}
\bibliography{reference.bib}

\end{document}